\documentclass[12pt]{article}    
\pdfoutput=1

\usepackage[vmargin={1.1in, 1.1in},hmargin={1in, 1in}]{geometry}

\usepackage{latexsym}
\usepackage{amssymb}
\usepackage{amsmath}
\usepackage{amsthm}
\usepackage{mathtools}
\usepackage{amsfonts}
\usepackage{graphicx}
\usepackage{enumerate}
\usepackage{color}
	\definecolor{MyDarkBlue}{rgb}{0,0.08,0.45}
\usepackage{hyperref}
	\hypersetup{linkcolor=MyDarkBlue, citecolor=MyDarkBlue, colorlinks=true}
\usepackage{setspace}
\usepackage[longnamesfirst]{natbib}
\newcommand\cites[1]{\citeauthor{#1}'s \citeyearpar{#1}}
\usepackage{hypernat}
\newtheorem*{theorem*}{Theorem}

\newtheorem{prop}{Proposition}

\theoremstyle{definition}

\newtheorem*{axiom*}{Axiom}

\newcommand{\finexhere}{\begingroup \let\mathqed\math@qedhere
    \let\qed@elt\setQED@elt \blacktriangle@stack\relax\relax \endgroup}

\newcommand{\cP}{\mathcal{P}}

\newcommand{\bE}{\mathbb{E}}

\newcommand{\bR}{\mathbb{R}}

\usepackage{xcolor}

\title{\vspace*{-3em}Hippocratic utility and Status Quo Bias}
\author{Tomasz Strzalecki\thanks{Department of Economics, Harvard University. Date: \today. I thank Elie Tamer, Davide Viviano, and Cheaheon Lim for helpful comments.}}
\date{}
\begin{document}

	\maketitle
\begin{abstract}
A utility function has been proposed that values more lives that are lost than those that are saved. I do not dispute the ethical motivation behind this kind of asymmetry. However, I show with a simple example that the scope of applicability of such a decision criterion is considerably more limited than it may first appear.
\end{abstract}

\section{A Decision Criterion}
Imagine the following scenario: a disease is present in the population and randomized controlled trials show that the ``old'' and default drug, developed in 1926, saves 10\%  of lives. According to randomized controlled trials, a ``new'' drug, just developed in 2026, saves 20\% of lives. A decision maker has to decide whether to give the default drug ($d=0)$ or the new drug ($d=1$) to a patient who has just developed the disease. 
 
This situation can be captured by the potential outcomes model \citep{neyman1923application,rubin1974estimating}. We divide the population into four principal strata. Each patient has deterministic responses $y=(y^0, y^1)$ to both drugs: $y^d=0$ if the patient dies and $y^d=1$ if they survive, for $d=0, 1$. The type of the patient is unknown to the decision maker, who only knows that $P(y^0=1)=10\%$ and $P(y^1=1)=20\%$. That is, only the marginals of the distribution are given, but the joint distribution $P$ over $\{0, 1\}^2$ is unknown. Let $\cP$ be the set of all probability distributions $P$ that are consistent with the marginals. This is the most information we can hope to obtain from a randomized controlled trial.



%


Suppose that the decision maker maximizes expected utility:
$\max_{d\in \{0, 1\}} \bE_P \ u(d, y^0, y^1),$
where $\bE_P$ is the expectation under probability $P$ and $u:D\times Y\times Y\to \bR$ is the utility function that depends on the decision $d\in D$ and the vector potential outcomes $y\in Y\times Y$.\footnote{Or equivalently, as is  often done in statistics, minimizes the loss function $\ell=-u$.}

Crucially, utility  depends on the whole vector of potential outcomes $(y^0, y^1)$. Such utility functions have recently been proposed in the medical decision making literature  \citep{ben2024policy,christy2024starting}. Motivated by the Hippocratic Oath  ``first, do no harm,'' they propose what I will call the \emph{Hippocratic utility function}, where the utility difference between $d=1$ and $d=0$ is given by:
\begin{equation}
u(d=1, y^0, y^1)-u(d=0, y^0, y^1)=\begin{cases}
	0 &\text{ if }y^0=y^1\\
	1 &\text{ if }y^0<y^1\\
	-\lambda &\text{ if }y^0>y^1,\\
\end{cases}	
\end{equation} 
%
%
where $\lambda>1$. That is, a death caused by the new drug is weighted more heavily than a death under the old drug. The parameter $\lambda>1$ measures the strength of the \emph{Hippocratic penalty}, which mathematically is the degree of the asymmetry in the utility function.

Formally this is similar to the \emph{loss aversion} of \cite{KT79}, except they were using this formalism to provide a descriptive theory of (biased) behavior, whereas we are now in the normative realm. Recognizing this distinction, \cite{mueller2023personalized} insist that for a policymaker we must have $\lambda=1$ but from an individual viewpoint the value of $\lambda$ can exceed one.

In Sections 2 and 3, I will briefly discuss two well-known problems with Hippocratic utility. Despite these arguments many scholars have a deep conviction that Hippocratic utility  is the normatively correct criterion. In Section 4, I will demonstrate another problem, independent of those two above. Even if you do not regard the issues in Sections 2 and 3 as decisive, I will argue that Hippocratic utility has a much smaller domain of applicability than is often presumed.

\section{Choosing the Dominated Option}
Much has been discussed about the fact that if $\lambda$ is high enough, then the old drug will be prescribed even though the new one saves twice as many lives. 

\begin{prop}
For any  $P\in \cP$ that puts a positive probability on all four strata there exists $\lambda>1$ such that the optimal decision is $d=0$.	
\end{prop}

\begin{proof}
    Let the probabilities of the four strata under $P$ be $P^{00} :=P(y^0=0, y^1=0), P^{01}:=P(y^0=0, y^1=1), P^{10}:=P(y^0=1, y^1=0),P^{11}:=P(y^0=1, y^1=1)$. The expected utility difference between $d=1$ and $d=0$ is $P^{01}-\lambda P^{10}$. We know that $P^{01}+P^{11}=20\%$ and $P^{10}+P^{11}=10\%$. Thus, the expected utility difference is $10\%+(1-\lambda)P^{10}$. This is less than zero for $\lambda$ large enough, provided that $P^{10}>0$.\footnote{This is bigger than zero for all values of $\lambda>1$ when $P^{10}=0$, which \cite{pearl2022probabilities} calles the monotonicity assumption; this is related but distinct from the monotonicity assumption of \cite{ImbensGuidoW1994IAEO}.}	
\end{proof}

This violation of stochastic dominance arises because Hippocratic utility penalizes deaths differently across strata. Proponents of Hippocratic utility argue that this is precisely their intention and that weighting some lives more heavily than others in this way is ethically well‑motivated.

Based on this violation of dominance, \cite{dawid2023personalised} ``argue that this approach is dangerously misguided and should not be
used in practice.'' Such objections have also been voiced by \cite{sarvet2025perspectives}. \cite{gelman2025russian} argue that utility should be defined directly as a function of the marginals of $P$ (what is sometimes called ``stochastic potential outcomes''), and not as a $P$-expectation of a function defined over $y$. This rules out Hippocratic utility because it leads to decision $d=1$.

\section{Choice Indeterminacy}
It has also been noticed that if only the marginals of $P$ are known, the decision criterion does not offer any clear advice to the decision maker. As we showed in the example above, not all probability measures $P\in \cP$ lead to the same decision because expected utility depends on the correlation structure of $P$. In our example, if $P\in \cP$ and $P^{10}=0$, then no matter how high $\lambda$ is, it is optimal to choose $d=1$, while for all other $P\in \cP$ the decision $d=0$ is taken for high enough values of $\lambda$. 

Thus, the decision will depend on a ``free parameter.'' To proceed, we could choose $P\in \cP$ to minimize regret, as in  \cite{ben2024policy}. Or we could instead use maxmin expected utility: $\max_{d\in \{0, 1\}} \min_{P\in \cP} \bE_P \ u(d, y^0, y^1).$ We could also assume that $P$ is a product measure. Or we could use maxmax utility. There are infinitely many possibilities\ldots

Each of these ways of choosing the free parameter leads to a different decision rule, reflecting a different normative stance. And the choice between these rules is not based on the data at hand, and not even on the value of $\lambda$, but rather on some other reasoning. 

Another solution to this problem is to only use utility functions that lead to the same decision for all $P\in \cP$. \cite{koch2025statistical} show that the value of expected utility depends only on the marginals if and only if the utility is an additive function of the vector of potential outcomes. This rules out Hippocratic utility (and coincides with \cites{gelman2025russian} proposal in the binary case).

The choice indeterminacy is a consequence of the fact that $P$ is only partially identified. This is related to the doubts expressed by \cite{sarvet2025perspectives}  who say ``a serious problem is that we have no direct evidence that these principal strata exist. '' As discussed by \cite{mueller2023personalized,mueller2025meaning}, fusing experimental data with observational data can help shrink the set $\cP$ and in some such cases the choice indeterminacy disappears (see also \citealt{sarvet2025rejoinder}, who show that in the case of fused data what they call ``counterfactual'' and ``interventionist'' decision rules coincide). I will not delve into this issue because my argument against Hippocratic utility (in the next section) applies even in cases where we somehow magically know $P$. 

\section{Status Quo Bias}
My thought experiment corresponds to decision problem 1 in Table \ref{tab:problems}, where  $d=0$ is one chemical compound (call it $A$) and $d=1$ is another chemical compound (call it $B$).

\begin{table}[htb!]
	\centering
	\begin{tabular}{c|cc}
		\hline
		decision problem & $d=0$ & $d=1$ \\
		\hline
			           1 &              compound $A$ &              compound $B$\\
			           2 &              compound $B$ &              compound $A$\\
			           3 &              do not treat &              compound $B$\\
			           4 &              do not treat, patient self-medicates with $A$ &              compound $B$\\
					   
		\hline
		\end{tabular}
		\caption{Decision Problems}	\label{tab:problems}

\end{table}

Consider now decision problem 2, where compound $B$ was discovered in 1926 and compound $A$ was discovered in 2026. Now a decision maker who has  $\lambda$ high enough will prescribe compound $B$ to the patient. Thus, by comparing decision problems 1 and 2, we conclude that Hippocratic utility suffers from status quo bias \citep{samuelson1988status}: the decision whether to prescribe  compound $A$ or compound $B$ depends on a completely random historical fact (which compound was discovered first). 

I have a hard time believing that a decision maker who wants to do ``no harm'' would want to base their decisions on random historical facts. I hope that the reader agrees with me here! I am not arguing against using Hippocratic utility in general, just that it does not apply to decision problems 1 and 2.\footnote{We might want to privilege the established compound because, relatively speaking, we have far fewer trials using the new one, or because we fear that the new compound may have negative long-term side effects not yet revealed by randomized trials. These are important considerations and should be modeled explicitly. In this note I am focusing on the \emph{ethical} reasons for taking $\lambda>1$, so I am abstracting from such scientific uncertainty.}

To be fair to the literature, the proponents of the utility function with $\lambda>1$ focus on slightly different situations than the one described here. Consider decision problem 3, where choosing $d=0$ corresponds to not giving a drug at all. When $d = 0$ corresponds to ``no intervention,'' there is perhaps a case to be made. When $d=0$ is ``whichever drug was discovered first,'' there isn't a case in my view.



But what does ``no intervention'' really mean? Consider decision problem 4, which is a slight modification of decision problem 3. Now compound $A$ is sold over-the-counter. The decision maker thinks that if decision  $d=0$ ``no treatment at all'' is taken, the patient will start self-medicating with compound $A$  with some probability. If that probability is high, the effective value of $\lambda$ should be close to one, so the same decision is made as with $\lambda=1$. Thus, I hope that if you like Hippocratic utility because of your ethical conviction, you agree with me that we should choose $d=1$ in decision problem 4. But then what is exactly the difference between problems 3 and 4 that makes you choose $d=0$ in problem 3?

\section{Discussion}

In principle, there is a way to apply Hippocratic utility to decision problems 1 and 2 without introducing status quo bias. Here, the researcher would define the utility function directly over  compounds $A$ and $B$, not objects like $d=0$ and $d=1$. This leads to internally consistent choices that satisfy completeness and transitivity, but violate dominance. The only difference from before is that now the decision maker prefers compound $A$ over compound $B$ independently of which one is the ``default action'' or which one was discovered first. Now we are sacrificing lives based on the chemical difference between compounds $A$ and $B$.\footnote{Of course, we could violate dominance because of the  production costs of different compounds. But those costs do not depend on strata, so they cannot lead to Hippocratic utility.} See \cite{koch2026axiomaticfoundationdecisionscounterfactual} for a formal treatment.





There are also legal arguments for using Hippocratic utility. When a treated patient dies, it may be easier for their family to sue the hospital, compared to the family of a dead, untreated patient \citep{tian2000probabilities,mueller2025meaning}. This motivation for Hippocratic utility is distinct from the ethical one, because here we would choose $d=0$ in both decision problems 3 and 4. In such applications the value of $\lambda$ should be proportional to the cost of the lawsuit (relative to the moral value of saving a life).


I conclude by proposing that we clearly separate two issues: (1) what is statistically identifiable, and (2) what the utility function should look like.  Assuming that we know $P$ exactly may be unrealistic, but it allows us to zoom in on the structure of the utility function. Because this is something that researchers disagree about, having clarity might be useful.
   
\begin{spacing}{1.25}
\bibliographystyle{econometrica}
\bibliography{bibilio-SDT.bib}
\end{spacing}

\end{document}